\documentclass[journal]{IEEEtran}
\usepackage{amsmath,amssymb,boxedminipage,subfigure}

\newtheorem{definition}{Definition}[section]
\newtheorem{lemma}{Lemma}[section]
\newtheorem{theorem}{Theorem}[section]
\newenvironment{proof}{\vspace{8pt}
\noindent{\em Proof}: }{{\hfill {\large $\Box$}} \vspace{8pt}}

\ifCLASSINFOpdf
\else
\fi
\begin{document}

\title{Upper Bounds on Matching Families in $\mathbb{Z}_{pq}^n$}

\author{Yeow~Meng~Chee,~\IEEEmembership{Senior~Member,~IEEE,}
        San~Ling,~Huaxiong~Wang,~Liang~Feng~Zhang
\thanks{Manuscript received September 29, 2012; revised February 09, 2013;
accepted March 03, 2013.}
\thanks{
The authors are with the Division of Mathematical Sciences, School of Physical
and Mathematical Sciences, Nanyang Technological University, Singapore
637371 (e-mail: ymchee@ntu.edu.sg; lingsan@ntu.edu.sg; hxwang@ntu.edu.sg; liangf.zhang@gmail.com).
}
\thanks{Communicated by V. Guruswami, Associate Editor,
IEEE Transactions on Information Theory.}
\thanks{Copyright (c) 2012 IEEE. Personal use of this material is permitted.  However, permission to use this material for any other purposes must be obtained from the IEEE by sending a request to pubs-permissions@ieee.org.}
}

\markboth{
IEEE~TRANSACTIONS~ON~INFORMATION~THEORY,~VOL.~xx,~NO.~x,~~~ 2013}%
{Shell \MakeLowercase{\textit{et al.}}:
On the Largest Size of Matching Families in $\mathbb{Z}_{pq}^n$}

\maketitle

\begin{abstract}
\textit{Matching families} are one of the major ingredients in the construction of
{\em locally decodable codes}
(LDCs) and the best known constructions of LDCs with a
constant number of queries are based on matching families.
The determination of the largest size of  any  matching family in $\mathbb{Z}_m^n$, where $\mathbb{Z}_m$
is the ring of integers modulo $m$, is an interesting problem.
In this paper, we show an upper bound of $O((pq)^{0.625n+0.125})$ for
 the size of any matching family in
$\mathbb{Z}_{pq}^n$, where  $p$ and $q$ are two distinct primes. Our  bound is
valid when  $n$ is a constant,  $p\rightarrow \infty$
and $p/q\rightarrow 1$. Our result  improves
an upper bound of Dvir {\it et al.}
\end{abstract}


\IEEEpeerreviewmaketitle

\section{Introduction}
\label{sec:introduction}

\IEEEPARstart{L}{\sc ocally Decodable Codes.} A classical error-correcting code $C$  allows one to encode any
message ${\bf x}=({\bf x}(1), \ldots, {\bf x}(k))$ of $k$ symbols  as a codeword $C({\bf x})$ of $N$ symbols  such that the message can be recovered even if
$C({\bf x})$ gets corrupted in a number of coordinates.
However, to recover even  a small fraction
of the message, one has to consider all or most of the coordinates of the
codeword. In such a scenario, more efficient  schemes are possible. They are known as
 {\em locally decodable codes} (LDCs).
Such codes allow the reconstruction of  any symbol of the message by
 looking at  a small number of coordinates of the codeword, even if a constant fraction of the codeword has been
  corrupted.

  Let $k,N$ be positive integers and let $\mathbb{F}$ be a finite field.
  For any ${\bf y}, {\bf z}\in \mathbb{F}^N$, we denote by
  $d_H( {\bf y}, {\bf z})$ the {\em Hamming distance}  between ${\bf y}$
 and ${\bf z}$.
 \begin{definition}
{\em (Locally Decodable Code)}
A code $C:\mathbb{F}^k\rightarrow \mathbb{F}^N$ is said to be
$(r,\delta,\epsilon)$-{\em locally decodable} if there is a randomized decoding algorithm
$D$  such that
\begin{enumerate}
\item for every ${\bf x}\in \mathbb{F}^k, i\in[k]$ and ${\bf y}\in \mathbb{F}^N$ such that $d_H(C({\bf x}),{\bf y})\leq \delta N$, $\Pr[D^{\bf y}(i)={\bf x}(i)]>1-\epsilon$,
where the probability is taken over the random coins of   $D$; and
\item $D$ makes at most $r$ queries to ${\bf y}$.
\end{enumerate}
\end{definition}
The efficiency of $C$ is measured by its  {\em query complexity} $r$ and {\em length} $N$ (as a function of $k$). Ideally, one would like
 both $r$ and $N$ to be as small as possible.

Implicit discussion of the notion of LDCs dates back to
\cite{BFLS91,Sud95,PS94}.
Katz and Trevisan \cite{KT00}
were the first to formally define LDCs and prove (superlinear) lower
bounds on their length.
Kerenidis and de Wolf \cite{KdW03} showed
a tight (exponential) lower bound for the length of 2-query LDCs.
Woodruff \cite{Woo07} obtained superlinear lower bounds for
 the length of   $r$-query LDCs, where $r\geq 3$.
More lower bounds for specific LDCs can be found in
\cite{GKST02,DJKRL02,Oba02,DS05,WW05,SL06}.
On the other hand, many constructions of  LDCs have been proposed in the past decade.
These constructions can be classified into three generations based on their technical ideas.
The  first-generation LDCs
\cite{BFLS91,KT00,BIK05,CGKS98}  are based on  (low-degree) multivariate polynomial
interpolation.
In such a code, each codeword
consists of evaluations of a low-degree polynomial in $\mathbb{F}[z_1,\ldots,z_n]$
at all points of $\mathbb{F}^n$, for some finite field $\mathbb{F}$.
The decoder recovers the value of the unknown polynomial at a point by shooting a line in
a random direction and decoding along it using noisy polynomial interpolation
\cite{BF90, Lip90, STV99}.
The second-generation LDCs \cite{BIKR02,WY05} are also based on low-degree multivariate polynomial interpolation but
with a clever use of recursion.
The third-generation LDCs, known as  {\em matching vector codes} (MV codes),
were initiated  by Yekhanin \cite{Yek07}
and developed further in \cite{Rag07,KY09,Efr09,Gop09,IS08,IS10,CFLWZ10,BET10,DGY10}.
The constructions involve novel combinatorial  and  algebraic ideas, where
 the key ingredient is the design of  large {\em matching families}  in
 $\mathbb{Z}_m^n$.
The interested reader may refer to Yekhanin  \cite{Yek12} for
 a good survey of  LDCs.

{\sc Matching Families.} Let $m$ and $n$ be positive integers.
For any vectors ${\bf u,v}\in \mathbb{Z}_m^n$, we denote by $\langle {\bf u,v} \rangle \triangleq
\sum_{i=1}^k {\bf u}(i){\bf v}(i) \bmod m$ their  {\em dot product}.
\begin{definition}\label{def:mf}
{\em (Matching Family)}
Let $S\subseteq \mathbb{Z}_m\setminus\{0\}$.
Two families  of vectors  $\mathcal{U}=\{{\bf u}_1,\ldots, {\bf u}_k\},~
\mathcal{V}=\{{\bf v}_1,\ldots, {\bf v}_k\}\subseteq \mathbb{Z}_m^n$ form an
 {\em $S$-matching family}  in $\mathbb{Z}_m^n$ if
\begin{enumerate}
\item $\langle {\bf u}_i, {\bf  v}_i\rangle=0$ for every $i\in [k]$; and
\item  $\langle {\bf u}_i, {\bf v}_j\rangle \in S$ for every $i,j\in [k]$ such that $i \neq j$.
\end{enumerate}
\end{definition}
The matching family defined above is of {\em size} $k$.
Dvir {\it et al.}~\cite{DGY10} showed that, if there is an $S$-matching family of
size $k$ in $\mathbb{Z}_m^n$, then there is an $(|S|+1)$-query LDC
encoding messages of length $k$ as codewords of length $m^n$.
Hence, large matching families are interesting because they result in short LDCs.
For any $S\subseteq \mathbb{Z}_m\setminus \{0\}$, it is  interesting  to
determine the largest size of any $S$-matching family in $\mathbb{Z}_m^n$.
When $S=\mathbb{Z}_m\setminus \{0\}$, this largest size is often denoted by
$k(m,n)$, which is clearly a {\em universal} upper bound for
the  size of any matching family in $\mathbb{Z}_m^n$.

{\sc Set Systems.}
The study of  matching families dates back to {\em set systems with restricted
 intersections}  \cite{BF98}, whose study was initiated in \cite{EKR61}.
\begin{definition}
{\em (Set System)}
Let $T$ and $S$ be two disjoint subsets of $\mathbb{Z}_m$. A collection
${\cal F}=\{F_1,\ldots,F_k\}$ of subsets of $[n]$ is said to be a
 $(T,S)$-{\em set system} over $[n]$ if
\begin{enumerate}
\item $|F_i|\bmod m\in T$ for every $i\in[k]$; and
\item $|F_i\cap F_j|\bmod m \in S$ for every $i,j\in [k]$ such that $i\neq j$.
\end{enumerate}
\end{definition}
The set system defined above is of {\em size} $k$. When $T=\{0\}$ and $S\subseteq \mathbb{Z}_m\setminus
\{0\}$, it is easy to show that the $(T,S)$-set system ${\cal F}$ yields
an  $S$-matching family of size $k$ in $\mathbb{Z}_m^n$. To see this,
let ${\bf u}_i={\bf v}_i\in \mathbb{Z}_m^n$ be  the characteristic vector of $F_i$ for every
$i\in [k]$, where  ${\bf u}_i(j)={\bf v}_i(j)=1$
for every $j\in F_i$ and 0 otherwise.
Clearly, ${\cal U}=\{{\bf u}_1,\ldots,{\bf u}_k\}$
and ${\cal V}=\{{\bf v}_1, \ldots, {\bf v}_k\}$ form an $S$-matching family of size $k$ in $\mathbb{Z}_m^n$.

When $m$ is a prime power  and $n\geq m$, Deza {\it et al.}~\cite{DFS83} and Babai {\it et al.}~\cite{BFKS01} showed that the largest size of any $(\{ 0 \},\mathbb{Z}_m\setminus\{0\})$-set systems over $[n]$ cannot be greater than
${n\choose m-1}+\cdots+{n\choose 0}$.
For any  integer $m$, Sgall  \cite{Sga99} showed that the  largest size of any
 $(\{ 0 \},\mathbb{Z}_m\setminus\{0\})$-set system over $[n]$
is bounded by $O(2^{0.5n})$.
On the other hand, Grolmusz \cite{Gro00}  constructed a $(\{ 0 \},\mathbb{Z}_m\setminus\{0\})$-set system  of  (superpolynomial) size $\exp(O((\log n)^r/(\log\log n)^{r-1}))$
over $[n]$ when  $m$  has $r\geq 2$ distinct prime divisors.
Grolmusz's set systems result in superpolynomial-sized matching families in $\mathbb{Z}_m^n$,  which have been  the key ingredient for
Efremenko's LDCs \cite{Efr09}.

{\sc Bounds.}
Due to the difficulty of determining $k(m,n)$ precisely, it is
 interesting to give both   lower and upper bounds  for $k(m,n)$.
When  $m\leq n$,
a {\em simple lower bound} for $k(m,n)$  is $k\triangleq {n\choose m-1}$.
To see this,  let $\mathcal{U}=\{{\bf u}_1,\ldots, {\bf u}_k\}$
be the set of all 0-1 vectors of
{\em Hamming weight} (i.e., the number of nonzero components) $m-1$ in $\mathbb{Z}_m^n$.
Let ${\bf v}_i=\textbf{1}-{\bf u}_i$ for every $i\in[k]$, where $\textbf{1}$
is the all-one vector. Then $\mathcal{U}$ and $\mathcal{V}=\{{\bf v}_1,\ldots, {\bf v}_k\}$
form a matching family  of size $k$.
When $m$ is a composite number with $r\geq 2$ distinct prime factors, the
 $(\{ 0 \},\mathbb{Z}_m\setminus\{0\})$-set systems of   \cite{Gro00,Kut02,DGY10}
result in superpolynomial-sized matching families in $\mathbb{Z}_m^n$.
In particular,
we have that $k(m,n)\geq \exp(O(\log^2n/\log\log n))$
when $m=pq$ for two distinct primes $p$ and $q$.
On the other hand,
Dvir {\it et al.}~\cite{DGY10} obtained upper bounds for $k(m,n)$ for various settings of the
integers $m$ and $n$. More precisely, they showed that
\begin{enumerate}
\item $k(m,n)\leq m^{n-1+o_m(1)}$ for any  integers $m$ and $n$, where $o_m(1)$ is a term that tends to 0 as $m$
approaches infinity;
\item $k(p,n)\leq \min\{1+{n+p-2\choose p-1}, 4p^{0.5n}+2\}$ for any prime $p$ and  integer $n$;
\item $k(m,n)\leq (m/q)^nk(q,n)$
for any  integers $m,n$ and $q$ such that $q|m$ and $\gcd(q,m/q)=1$.
\end{enumerate}
In particular, the latter two bounds imply that    $k(m,n)\leq p^n (4q^{0.5n}+2)$
when $m=pq$ for two distinct primes $p$ and $q$ such that $p\leq q$.

 {\sc Our Results.}
Dvir {\it et al.}~\cite{DGY10} conjectured that $k(m,n)\leq O(m^{0.5n})$ for any
integers $m$ and $n$.  A {\em special case} where the conjecture is open is when $n$ is
a constant,  and
 $m=pq$ for two distinct primes $p,q$ such that $p\rightarrow \infty$ and
 $p/q\rightarrow 1$.
 In this paper, we show that
 $k(m,n)\leq O(m^{0.625n+0.125})$ for this special case, which improves the
 best known upper bound that can be derived from results of Dvir {\it et al.} in \cite{DGY10}, i.e.,
 $k(m,n)\leq p^n (4q^{0.5n}+2)=O(m^{0.75n})$.

{\sc Our Techniques.}
Let $\mathcal{U}=\{{\bf u}_1,\ldots,{\bf u}_k\},\
 \mathcal{V}=\{{\bf v}_1,\ldots,{\bf v}_k\}\subseteq \mathbb{Z}_m^n$
be a matching family of size $k=k(m,n)$, where $m=pq$ for two distinct primes
 $p$ and $q$.
We say that ${\bf u},{\bf v}\in \mathbb{Z}_m^n$ are {\em equivalent}  (and write ${\bf u}\sim{\bf v}$) if there is a  $\lambda\in \mathbb{Z}_m^*$
such that ${\bf u}(i)=\lambda {\bf v}(i)$ for every $i\in[n]$,
where $\mathbb{Z}_m^*$ is the set of units of $\mathbb{Z}_m$.
Clearly, no two elements of
 $\mathcal{U}$ (resp. $\mathcal{V}$) can be  equivalent to each other.
Let  $s,t\in \{1,p,q,m\}$. We say that
 $({\bf u}_i,{\bf v}_i)$
is of {\em type $(s,t)$} if $\gcd({\bf u}_i(1),\ldots,{\bf u}_i(n),m)=s$ and
$\gcd({\bf v}_i(1),\ldots, {\bf v}_i(n),m)=t$.
We can partition the set $\{({\bf u}_i, {\bf v}_i): i\in [k]\}$ of pairs
according to their types.
Let $N_{s,t}$ be the number of pairs of type $(s,t)$.
Then we have the following observations:
\begin{enumerate}
\item $N_{s,t}\leq 1$ when $m|st$ (see Lemma \ref{lem:stdm});
\item  $N_{s,t}\leq k(q,n)$ when $(s,t)\in \{(1,p), (p,1), (p,p)\}$
(see Lemma \ref{lem:p_family}); and
\item $N_{s,t}\leq k(p,n)$ when $(s,t)\in \{(1,q), (q,1), (q,q)\}$ (see Lemma \ref{lem:q_family}).
\end{enumerate}
These observations in turn imply that
$
k\leq 9+N_{1,1}+3k(p,n)+3k(q,n)
$
and enable us to reduce the problem of upper-bounding $k$ to that of establishing an upper bound for $N_{1,1}$.

As in \cite{DGY10}, we establish an upper bound for $N_{1,1}$ by using
an interesting relation between matching families  and the {\em expanding
properties}  of the {\em projective graphs} (which will be explained shortly).
Let
\begin{equation}
\begin{split}
\mathbb{S}_{n,m}&=\{{\bf u}\in \mathbb{Z}_m^n: \gcd({\bf u}(1),\ldots ,{\bf u}(n),m)=1\}
{\rm ~and~}\\
\mathbb{P}_{n,m}&=\mathbb{H}_{n,m}= \mathbb{S}_{n,m}/\sim .
\end{split}
\end{equation}
We define the   {\em projective $(n-1)$-space} over $\mathbb{Z}_m$  to be the pair
$(\mathbb{P}_{n,m},
\mathbb{H}_{n,m})$.
We call the elements of
 $\mathbb{P}_{n,m}$   {\em points} and the elements of $\mathbb{H}_{n,m}$
{\em hyperplanes}.
We say that a point
${\bf u}$ {\em lies on} a hyperplane ${\bf v}$ if
$\langle {\bf u}, {\bf v}\rangle=0$.
The projective graph
 ${\bf G}_{n,m} $  is defined to be a bipartite graph with classes of vertices
 $\mathbb{P}_{n,m}\cup \mathbb{H}_{n,m}$,
where a point  ${\bf u}$ and a hyperplane ${\bf v}$
are {\em adjacent} if and only if ${\bf u}$ lies on ${\bf v}$. Vertices that are adjacent
to each other are called {\em neighbors}.
A  set ${\cal U}^\prime\subseteq \mathbb{P}_{n,m}$ has the {\em unique neighbor property} if,
 for every ${\bf u}\in {\cal U}^\prime$, there is a  hyperplane ${\bf v}$
such that ${\bf v}$ is adjacent to ${\bf u}$ but to no  other points  in ${\cal U}^\prime$
(see also Definition \ref{def:unp}).
Without loss of generality, let
$\{({\bf u}_i, {\bf v}_i): i\in [k^\prime]\}$ be the set of pairs of type $(1,1)$,
where $k^\prime=N_{1,1}$. Let $\mathcal{U}^\prime=\{{\bf u}_1,\ldots, {\bf u}_{k^\prime}\}
\subseteq \mathbb{P}_{n,m}$.
It is straightforward to see that ${\cal U}^\prime$ satisfies the unique neighbor
property (Lemma \ref{lem:unpm}).
For any $X\subseteq \mathcal{U}^\prime$, we denote by $N(X)$ the {\em neighborhood} of
$X$, i.e., the collection of vertices in $\mathbb{H}_{n,m}$ that are adjacent to some
vertex in $X$.
Since every point in ${\cal U}^\prime\setminus X$ must have a unique neighbor in
$\mathbb{H}_{n,m}\setminus N(X)$, we have that
\begin{equation}
|\mathcal{U}^\prime|\leq |X|+|\mathbb{H}_{n,m}|-|N(X)|.
\end{equation}

We show that ${\bf G}_{n,m}$ has some kind of {\em expanding property} (see Theorem \ref{thm:bound_NX}),
meaning that
$|N(X)|$ is  large for certain choices of  $X$, which
allows us to obtain the expected upper bound for $k^\prime=N_{1,1}$
(see Theorems \ref{thm:general_bound} and \ref{thm:upper_bound_m_infty}).
When $m$ is a prime, such an expanding  property of ${\bf G}_{n,m}$  was  proved
 by Alon  \cite{Alo86} using  the {\em spectral method} and it says that
\begin{equation}
\label{eqn:expanding}
|N(X)|\geq |\mathbb{P}_{n,m}|-|\mathbb{P}_{n,m}|^{n/(n-1)}/|X|,
\end{equation}
where $X\subseteq \mathbb{P}_{n,m}$ is arbitrary.

Let $A_{n,m}=(a_{{\bf u} {\bf v}})$ be the {\em adjacency matrix} of ${\bf G}_{n,m}$,
where the rows are labeled by the points, the columns by the hyperplanes, and
$a_{\bf uv}=1$ if and only if ${\bf u}$ and ${\bf v}$
are adjacent. Note that the matrix $A_{n,m}$ may take many different forms
because the sets $\mathbb{P}_{n,m}$ and $\mathbb{H}_{n,m}$ are not {\em ordered}.
However, from now on, we always assume that $\mathbb{P}_{n,m}$ and $\mathbb{H}_{n,m}$
 are identical to
each other as ordered sets. Hence, $A_{n,m}$ is {\em symmetric}.
 Let  $\chi_{\mbox{\tiny\itshape X}}$ be the characteristic vector of
$X$, where the components of $\chi_{\mbox{\tiny\itshape X}}$ are labeled by the elements  ${\bf u}\in \mathbb{P}_{n,m}$ and $\chi_{\mbox{\tiny\itshape X}}({\bf u})=1$ if
and only if ${\bf u}\in X$.
Alon \cite{Alo86} obtained both an upper bound and a lower bound  for $\chi_{\mbox{\tiny\itshape X}}^t B_{n,m} \chi_{\mbox{\tiny\itshape X}}$ that
jointly result in (\ref{eqn:expanding}),
 where $B_{n,m}=A_{n,m}A^t_{n,m}$ with the superscript $t$ denoting the {\em transpose} of a matrix.
More precisely,  Alon \cite{Alo86} determined the eigenvalues of $B_{n,m}$ and
 represented $\chi_{\mbox{\tiny\itshape X}}$ as a linear combination of the eigenvectors of $B_{n,m}$.
In this paper, we develop their spectral method further
and show a {\em tensor lemma} on
$B_{n,m}$ (see Lemma \ref{pro:tensorp}), which says that
${\bf G}_{n,m}$  is a tensor product of ${\bf G}_{n,p}$ and ${\bf G}_{n,q}$
when $m=pq$, where $p$ and $q$ are two distinct primes. As in \cite{Alo86}, we
determine the eigenvalues of $B_{n,m}$ and represent
$\chi_{\mbox{\tiny\itshape X}}$ as a linear combination of
the eigenvectors of $B_{n,m}$.
We obtain both an upper bound and a lower bound for $\chi_{\mbox{\tiny\itshape X}}^t B_{n,m} \chi_{\mbox{\tiny\itshape X}}$, which are then used to show that ${\bf G}_{n,m}$ has some kind of expanding
property (see Theorem \ref{thm:bound_NX}).

 {\sc Subsequent Work.}
Recently,  in a follow-up work, Bhowmick {\it et al.} \cite{BDL12} obtained new  upper bounds for $k(m,n)$.
They used different techniques and showed that $k(m,n)\leq m^{0.5n+14\log m}$ for any integers $m$ and $n$.
In particular, their upper bound translates into $k(m,n)\leq m^{0.5n+O(1)}$ for the
special case we consider in this paper.

{\sc Organization.}
In Section \ref{sec:pg}, we study projective graphs over $\mathbb{Z}_m$ and matrices
associated with such graphs. In Section \ref{mainresult-sec3}, we  establish our  upper bound for  $k(pq,n)$ using
 the unique neighbor property in projective graphs.
Section \ref{sec:con} contains some concluding remarks.

\section{Projective Graphs and  Associated Matrices}
\label{sec:pg}

Let $d$ be a positive integer. We denote by $\textbf{0}_d,~\textbf{1}_d,~I_d$ and $J_d$
 the all-zero (either row or column) vector of dimension $d$, all-one (either row or column) vector of dimension $d$, identity matrix of
 order $d$ and all-one matrix of order $d$, respectively.  We denote by $O$ an all-zero matrix
 whose size is clear from the context. We also define
\begin{equation}
\begin{split}
K_d&=I_d+J_d,\\
L_{d}&=
\begin{pmatrix}
(d+1)  I_{d}-J_{d} &  -\textbf{1}_d
\end{pmatrix},{\rm ~and}\\
R_d&=
\begin{pmatrix}
I_{d} &
-\textbf{1}_d
\end{pmatrix}^t .
\end{split}
\end{equation}
Let $A = (a_{ij})$
and $B$ be two matrices. We define their {\em tensor product} to be the block matrix
$A\otimes B = (a_{ij}\cdot B)$.
We say that  $A\simeq B$
if  $A$ can be obtained from $B$ by {\em simultaneously}  permuting the rows and columns
(i.e., apply the same permutation to both the rows and columns).
Clearly,  $A$ and $B$ have the same
eigenvalues if $A\simeq  B$.

In this section, we study the projective
graph ${\bf G}_{n,m}$ defined in Section \ref{sec:introduction}. We also follow the
notation there.
Let $\theta_{n,m}=|\mathbb{P}_{n,m}|$ be the number of points (or hyperplanes) in the projective $(n-1)$-space
over $\mathbb{Z}_m$.
Chee and Ling \cite{CL93} showed that
\begin{equation}
\label{equation:order}
\theta_{n,m}=m^{n-1} \prod_{p|m}
(1+1/p+
\cdots+1/p^{n-1})
\end{equation}
and $
|N({\bf u})|=|N({\bf v})|=\theta_{n-1,m}
$  for every point ${\bf u}$ and hyperplane  ${\bf v}$.
When $m$ is  prime,  Alon \cite{Alo86} showed that
$\theta_{n-1,m}^2$ is an eigenvalue of $B_{n,m}$ of multiplicity 1 and
$m^{n-2}$  is an eigenvalue of $B_{n,m}$ of multiplicity $\theta_{n,m}-1$.
Furthermore, an eigenvector of $B_{n,m}$  with eigenvalue $\theta_{n-1,m}^2$ is
 $\textbf{1}$
 and linearly independent eigenvectors of $B_{n,m}$  with
eigenvalue $m^{n-2}$ can be chosen to be the  columns of  $R_d$, where $d=\theta_{n,m}-1$.
However, the  eigenvalues of $B_{n,m}$
have not been studied   when $m$ is  composite.
Here, we determine the eigenvalues  of
$B_{n,m}$ when $m=pq$ for two distinct primes $p$ and $q$.
\begin{lemma}\label{pro:tensorp}
{\em  (Tensor Lemma)}
Let $n>1$ be an integer and let $m=p q $ for two distinct primes $p $ and $q$.
Then
$B_{n,m} \simeq B_{n,p }\otimes
B_{n,q }$.
\end{lemma}

\begin{proof}
Let $\pi: \mathbb{P}_{n,p }\times \mathbb{P}_{n,q }
\rightarrow \mathbb{P}_{n,m}$
be the mapping defined by $\pi({\bf u},{\bf v})={\bf w}$,
where
\begin{equation}
\label{eqn:defw}
\begin{split}
{\bf w}(i)\equiv {\bf u}(i)\bmod p \quad {\rm and} \quad
{\bf w}(i)\equiv {\bf v}(i)\bmod q
\end{split}
\end{equation}
 for every $i\in[n]$. Then $\pi$ is well-defined.
To see this, let ${\bf w}^\prime =\pi({\bf u}^\prime, {\bf v}^\prime)$ and ${\bf w}=\pi({\bf u}, {\bf v})$
for  ${\bf u}, {\bf u}^\prime\in \mathbb{S}_{n,p }$ and $
{\bf v}, {\bf v}^\prime\in \mathbb{S}_{n,q }$.  If
${\bf u}\sim {\bf u}^\prime$ and ${\bf v}\sim {\bf v}^\prime$, then there are integers
$\lambda\in \mathbb{Z}_{p }^*$ and $\mu\in \mathbb{Z}_{q }^*$ such that
\begin{equation}
\label{eqn:uupvvp}
\begin{split}
{\bf u}^\prime(i)\equiv \lambda {\bf u}(i)\bmod p \quad {\rm and} \quad
{\bf v}^\prime(i)\equiv \mu {\bf v}(i)\bmod q
\end{split}
\end{equation}
for every
$i\in [n]$. Let $\delta\in \mathbb{Z}_m^*$ be an integer such that
\begin{equation}
\label{eqn:defdelta}
\begin{split}
\delta\equiv \lambda \bmod p \quad {\rm and} \quad
\delta \equiv \mu \bmod q.
\end{split}
\end{equation}
By (\ref{eqn:defw}), (\ref{eqn:uupvvp}) and (\ref{eqn:defdelta}), we have that  ${\bf w}^\prime (i)\equiv \delta  {\bf w}(i) \bmod m$
for every $i\in [n]$. Hence, ${\bf w}\sim {\bf w}^\prime$.

Let $\mathbb{P}_{n,p }=\{{\bf u}_1,\ldots, {\bf u}_{\ell_1}\}$ and
$\mathbb{P}_{n,q }=\{{\bf v}_{1},\ldots, {\bf v}_{\ell_2}\}$, where $\ell_1=\theta_{n,p}$ and $\ell_2=\theta_{n,q}$.
It is clear that $\pi$ is injective and $\theta_{n,m}=\ell_1\ell_2$
(this is clear from (\ref{equation:order})). It follows that
  $\pi$ is bijective and
\begin{equation}\label{equ:ordering}
\mathbb{P}_{n,m}=
 \{\pi({\bf u}_{1},{\bf  v}_{1}),\ldots, \pi({\bf u}_{1},{\bf  v}_{\ell_2}),\ldots, \pi({\bf u}_{\ell_1},
 {\bf v}_{\ell_2})\}.
\end{equation}
Let ${\bf w}$ and ${\bf w}^\prime$ be as above.  Then
$\langle {\bf w}, {\bf w}^\prime\rangle\equiv 0\bmod m$ if and only if
$\langle {\bf u}, {\bf u}^\prime\rangle \equiv 0\bmod p$ and $
\langle {\bf v}, {\bf v}^\prime\rangle \equiv 0\bmod q$.
Hence, the
$({\bf w}, {\bf w}^\prime)$-entry of $A_{n,m}$ is equal to 1
if and only if the $({\bf u}, {\bf u}^\prime)$-entry of $A_{n,p }$
and the $({\bf v},{\bf v}^\prime)$-entry of $A_{n,q }$ are both equal to 1.
Hence,  $A_{n,m}\simeq  A_{n,p }
\otimes A_{n,q }$. It follows that
\begin{equation*}
\begin{split}
B_{n,m}&=
A_{n,m}A_{n,m}^t \\
&\simeq  (A_{n,p}\otimes A_{n,q})(A_{n,p}\otimes A_{n,q})^t\\
&=
(A_{n,p}A_{n,p}^t)\otimes(A_{n,q} A_{n,q}^t)\\
&=B_{n,p}\otimes B_{n,q},
\end{split}
\end{equation*}
as desired.
\end{proof}

In fact, we could have concluded that $A_{n,m}=A_{n,p}\otimes A_{n,q}$ and therefore
$B_{n,m}=B_{n,p}\otimes B_{n,q}$ in Lemma \ref{pro:tensorp}. The sole reason that we
did not do so is that those matrices may take different forms, as noted in Section
\ref{sec:introduction}.
To facilitate further analysis, we make the
matrices unique
 such that
$A_{n,m}=A_{n,p}\otimes A_{n,q}$.
This can be achieved by making the sets $\mathbb{P}_{n,p}, \mathbb{P}_{n,q}$
and $\mathbb{P}_{n,m}$ unique.
To do so, we first make  $\mathbb{P}_{n,p}=
[{\bf u}_1,\ldots, {\bf u}_{\ell_1}]$ and
$\mathbb{P}_{n,q}=
[{\bf v}_1,\ldots, {\bf v}_{\ell_2}]$  unique as ordered sets, where $\ell_1=\theta_{n,p}$ and
$\ell_2=\theta_{n,q}$.
\begin{figure*}
\begin{center}
Figure 1: Ordered Point Sets

\vspace{0.25cm}

\begin{tabular}{|c|c|ccccccc|}
\hline
$\mathbb{P}_{3,2}$ & $\mathbb{P}_{3,3}$ & \multicolumn{7}{c|}{$\mathbb{P}_{3,6}$}  \\
\hline
(0, 0, 1) & (0, 0, 1) & (0, 0, 1) & (0, 3, 4) & (0, 3, 1) & (3, 0, 4) & (3, 0, 1) & (3, 3, 4) & (3, 3, 1)\\
(0, 1, 0) & (0, 1, 0) & (0, 4, 3) & (0, 1, 0) & (0, 1, 3) & (3, 4, 0) & (3, 4, 3) & (3, 1, 0) & (3, 1, 3)\\
(0, 1, 1) & (0, 1, 1) & (0, 4, 1) & (0, 1, 4) & (0, 1, 1) & (3, 4, 4) & (3, 4, 1) & (3, 1, 4) & (3, 1, 1)\\
(1, 0, 0) & (0, 1, 2) & (0, 4, 5) & (0, 1, 2) & (0, 1, 5) & (3, 4, 2) & (3, 4, 5) & (3, 1, 2) & (3, 1, 5)\\
(1, 0, 1) & (1, 0, 0) & (4, 0, 3) & (4, 3, 0) & (4, 3, 3) & (1, 0, 0) & (1, 0, 3) & (1, 3, 0) & (1, 3, 3)\\
(1, 1, 0) & (1, 0, 1) & (4, 0, 1) & (4, 3, 4) & (4, 3, 1) & (1, 0, 4) & (1, 0, 1) & (1, 3, 4) & (1, 3, 1)\\
(1, 1, 1) & (1, 0, 2) & (4, 0, 5) & (4, 3, 2) & (4, 3, 5) & (1, 0, 2) & (1, 0, 5) & (1, 3, 2) & (1, 3, 5)\\
          & (1, 1, 0) & (4, 4, 3) & (4, 1, 0) & (4, 1, 3) & (1, 4, 0) & (1, 4, 3) & (1, 1, 0) & (1, 1, 3)\\
          & (1, 1, 1) & (4, 4, 1) & (4, 1, 4) & (4, 1, 1) & (1, 4, 4) & (1, 4, 1) & (1, 1, 4) & (1, 1, 1)\\
          & (1, 1, 2) & (4, 4, 5) & (4, 1, 2) & (4, 1, 5) & (1, 4, 2) & (1, 4, 5) & (1, 1, 2) & (1, 1, 5)\\
          & (1, 2, 0) & (4, 2, 3) & (4, 5, 0) & (4, 5, 3) & (1, 2, 0) & (1, 2, 3) & (1, 5, 0) & (1, 5, 3)\\
          & (1, 2, 1) & (4, 2, 1) & (4, 5, 4) & (4, 5, 1) & (1, 2, 4) & (1, 2, 1) & (1, 5, 4) & (1, 5, 1)\\
          & (1, 2, 2) & (4, 2, 5) & (4, 5, 2) & (4, 5, 5) & (1, 2, 2) & (1, 2, 5) & (1, 5, 2) & (1, 5, 5)\\
\hline
\end{tabular}
\end{center}
\label{fig:order}
\end{figure*}

For example, as shown in Figure 1, we may set
$\mathbb{P}_{3,2}=[
(0, 0, 1),~(0, 1, 0),~ (0, 1, 1),~ (1, 0, 0),~ (1, 0, 1),~  (1, 1, 1)]$
and
$\mathbb{P}_{3,3}~=[
(0, 0, 1),(0, 1, 0), (0, 1, 1), (0, 1, 2), (1, 0, 0),  (1, 0, 1),\\ (1, 0, 2),
(1, 1,0), (1, 1,1), (1, 1,2), (1, 2,0), (1, 2, 1), (1, 2,2) ]$. Then
both $\mathbb{P}_{3,2}$ and $\mathbb{P}_{3,3}$ have been made unique as ordered sets.
(Here, each equivalence class in $\mathbb{P}_{n,p}$ is represented by the first element
when its elements are arranged in lexicographical order, and these representatives
are subsequently also arranged in lexicographical order.)
Once  $\mathbb{P}_{n,p}$ and $\mathbb{P}_{n,q}$
have been made unique as ordered sets, we can simply set  $\mathbb{P}_{n,m}
=[{\bf w}_1,{\bf w}_2, \ldots, {\bf w}_{\ell}]=
[\pi({\bf u}_1,{\bf v}_1), \pi({\bf u}_1,{\bf v}_2),\ldots,
\pi({\bf u}_{\ell_1},{\bf v}_{\ell_2})]$, where $\ell=\ell_1\ell_2$ and
${\bf w}_1=\pi({\bf u}_1,{\bf v}_1), {\bf w}_2=\pi({\bf u}_1,{\bf v}_2),\ldots,
{\bf w}_\ell=\pi({\bf u}_{\ell_1},{\bf v}_{\ell_2})$.
For example, as shown in Figure 1, $\mathbb{P}_{3,6}$ consists of
$\ell_1(=7)$ columns and the $i$th column corresponds to $\pi({\bf u}_i,{\bf v}_1),\ldots,
\pi({\bf u}_i,{\bf v}_{\ell_2})$ for every $i\in [\ell_1]$.
From now on, we suppose that the point sets $\mathbb{P}_{n,p}, \mathbb{P}_{n,q}$
and $\mathbb{P}_{n,m}$ have always been made unique, such as in the way illustrated above. Then we have
\begin{equation}
\label{eqn:tensorp}
A_{n,m}=A_{n,p}\otimes A_{n,q} \quad {\rm and} \quad
B_{n,m}=B_{n,p}\otimes B_{n,q}.
\end{equation}
 Let
$d_1=1, ~d_2=\ell_1-1,~d_3=\ell_2-1$ and $d_4=(\ell_1-1)(\ell_2-1)$. We define
an $\ell\times \ell$ matrix
\begin{equation}
\label{eqn:defY}
\begin{split}
Y&=
\begin{pmatrix}
Y_1 & Y_2 &Y_3 &Y_4
\end{pmatrix}\\
&=\begin{pmatrix}
\textbf{1}_\ell &  R_{d_2}
  \otimes \textbf{1}_{\ell_2}&  \textbf{1}_{\ell_1}  \otimes
R_{d_3} &  R_{d_2}  \otimes
R_{d_3}
\end{pmatrix}.
\end{split}
\end{equation}

\begin{lemma}\label{lemma:eigenvalue}
For every $s\in \{1,2,3,4\}$, the $d_s$ columns of $Y_s$ are linearly independent
eigenvectors of $B_{n,m}$ with eigenvalue $\lambda_s$, where
$\lambda_1 =\theta_{n-1,m}^2,~\lambda_2 =p^{n-2}\theta_{n-1,q}^2,~
\lambda_3 =  q^{n-2}\theta_{n-1,p}^2$ and $\lambda_4 =  m^{n-2}$.
\end{lemma}

\begin{proof}
The proof consists of simple verification. For example, when $s=4$, we have that
$
B_{n,m}\cdot Y_4=(B_{n,p}\otimes B_{n,q})\cdot (R_{d_2}\otimes R_{d_3})=
\left(B_{n,p}\cdot
R_{d_2}\right)
\otimes  \left(
B_{n,q}\cdot
R_{d_3}
\right)
=
\left(p^{n-2}\cdot
R_{d_2}
\right)
\otimes
\left(q^{n-2}\cdot
R_{d_3}
\right)
=\lambda_4 \cdot Y_4
$, where the first equality is due to (\ref{eqn:tensorp}).
Similarly, we can verify for $s\in \{1,2,3\}$. The linear independence
of the columns of $Y_s$ can be checked using the linear independence of
the columns of $R_{d_2}$ and $R_{d_3}$.
\end{proof}

\begin{lemma}
\label{lemma:eigenvector}
We have that
\begin{equation*}
\begin{split}
Y^{-1}&=
\ell^{-1}\cdot
\begin{pmatrix}
\textbf{1}_\ell \\
   L_{d_2} \otimes \textbf{1}_{\ell_2}\\
\textbf{1}_{\ell_1}\otimes L_{d_3}\\
L_{d_2}\otimes L_{d_3}
\end{pmatrix}
{\rm~and~} \\
Y^t\cdot Y&=
\begin{pmatrix}
\ell  & O & O & O\\
O & \ell_2K_{d_2} & O & O\\
O & O & \ell_1K_{d_3} & O\\
O & O & O & K_{d_2}\otimes K_{d_3}
\end{pmatrix}.
\end{split}
\end{equation*}
\end{lemma}

\begin{proof}
Note that $L_d\cdot R_d=(d+1)\cdot I_d$ and $\textbf{1}_d\cdot R_d=O$ and
 $R_d^t\cdot R_d=K_d$ for   every integer $d$.
Both equalities follow from simple calculations.
\end{proof}

\section{Main Result}\label{mainresult-sec3}

In this section, we present our main result, i.e., a new upper bound for $k(m,n)$, where $m=pq$ is the
product of two distinct primes $p$ and $q$.
As noted in the Section \ref{sec:introduction}, our arguments consist of a series of reductions. First of all,
we reduce the problem of finding an upper bound for $k(m,n)$ to
one of establishing an upper bound for $N_{1,1}$, the number of pairs
$({\bf u}_i, {\bf v}_i)$ of type $(1,1)$.   The latter problem is in turn reduced to the study of
the projective graph ${\bf G}_{n,m}$. More precisely, we follow the techniques of
\cite{DGY10} and use the unique neighbor property of ${\bf G}_{n,m}$. However,
the validity  of the technique  depends on
some expanding property of ${\bf G}_{n,m}$.

\subsection{An Expanding Property}\label{subsec:bound_pq}

We follow the notations of Section \ref{sec:pg}.
In this section, we  show that the projective
graph ${\bf G}_{n,m}$ has some kind of expanding property (see Theorem \ref{thm:bound_NX}), in
the sense that $|N(X)|$ is large for certain choices of $X$. Expanding properties of
the projective graph ${\bf G}_{n,p}$, where $p$ is a prime, have been studied by Alon \cite{Alo86} using the well-known
spectral method. In Section \ref{sec:pg}, we made the observation that
the  graph ${\bf G}_{n,m}$ is a tensor product of the graphs
${\bf G}_{n,p}$ and ${\bf G}_{n,q}$.
This observation enables us to obtain interesting properties
(see Lemmas \ref{lemma:eigenvalue}
and \ref{lemma:eigenvector})  which in turn  facilitate our proof
that  ${\bf G}_{n,m}$ has some kind of expanding property.

Let $\mathbb{N}$ be the set of nonnegative integers and let $\mathbb{R}$ be the field of real numbers.
For any vectors $\phi=(\phi_1,\ldots, \phi_\ell)^t,
\psi=(\psi_1,\ldots, \psi_\ell)^t \in \mathbb{R}^\ell$, we let
 $\langle  \phi,\psi\rangle=\sum_{i=1}^\ell \phi_i\cdot \psi_i$ and
 $\|\phi\|^2=\langle \phi,\phi\rangle$.
Furthermore, we define the weight of $\phi$ to be ${\rm wt}(\phi)=\sum_{i=1}^\ell \phi_i$.
For a set $X\subseteq \mathbb{P}_{n,m}$, we denote by $\chi_{\mbox{\tiny\itshape X}} \in \mathbb{R}^\ell$
its characteristic vector whose
components are labeled by the elements ${\bf u}\in \mathbb{P}_{n,m}$ and
$\chi_{\mbox{\tiny\itshape X}} ({\bf u})=1$ if ${\bf u}\in X$ and 0 otherwise.
 Due to
Lemmas  \ref{lemma:eigenvalue} and \ref{lemma:eigenvector},
the column vectors  of $Y$ form a basis of the vector space $\mathbb{R}^\ell$.
Therefore, there is a real vector
\begin{equation*}
\begin{split}
\alpha&=
\begin{pmatrix}
\alpha_1\\
 \alpha_2\\
  \alpha_3\\
   \alpha_4
\end{pmatrix}
\end{split},
\end{equation*}
where
\begin{equation*}
\begin{split}
\alpha_1=\alpha_{11},~
\alpha_2=
\begin{pmatrix}
\alpha_{21}\\
 \vdots\\
  \alpha_{2d_2}
\end{pmatrix},~
\alpha_3=
\begin{pmatrix}
\alpha_{31}\\
 \vdots\\
  \alpha_{3d_3}
\end{pmatrix},~
\alpha_4=
\begin{pmatrix}
\alpha_{41}\\
 \vdots\\
  \alpha_{4d_4}
\end{pmatrix}
\end{split}
\end{equation*}
such that $\chi_{\mbox{\tiny\itshape X}} $ can be written as a linear combination of the columns of $Y$, say
\begin{equation}
\label{eqn:chi}
\chi_{\mbox{\tiny\itshape X}} =Y\alpha=\sum_{s=1}^4 Y_s \alpha_s.
\end{equation}
Let $\psi=A_{n,m}^t   \chi_{\mbox{\tiny\itshape X}} $.
The main idea of Alon's spectral method in \cite{Alo86} is to establish both a lower
 bound and an upper bound for the following number:
\begin{equation}
\label{eqn:normpsi}
\begin{split}
\|\psi\|^2=\chi_{\mbox{\tiny\itshape X}}^t \cdot B_{n,m} \chi_{\mbox{\tiny\itshape X}}
&=\sum_{r=1}^4  \alpha_r^t  Y_r^t \cdot \sum_{s=1}^4 \lambda_s Y_s \alpha_s\\
&=\sum_{s=1}^4\lambda_s \|Y_s \alpha_s\|^2,
\end{split}
\end{equation}
where the second equality is due to  Lemma \ref{lemma:eigenvalue},
and the third equality follows from the second part of Lemma \ref{lemma:eigenvector}.
For  every  $s\in\{1,2,3,4\}$, we set
\begin{equation}
\label{eqn:defDs}
\Delta_s=\|Y_s\alpha_s\|^2.
\end{equation}

\begin{lemma}
\label{lemma:caldelta}
The quantities $\Delta_1$, $\Delta_2$ and $\Delta_3$ can be written as: 
\begin{equation}
\label{eqn:caldelta}
\begin{split}
\Delta_1=\ell \alpha_{11}^2, \quad
\Delta_2&= \ell_2 (\|\alpha_2\|^2+{\rm wt}(\alpha_2)^2) \quad {\rm and} \quad \\
\Delta_{3}&=\ell_1(\|\alpha_3\|^2+{\rm wt}(\alpha_3)^2).
\end{split}
\end{equation}
\end{lemma}

\begin{proof}
Lemma \ref{lemma:eigenvector} shows that
$Y_2^tY_2=\ell_2K_{d_2}$. Then we have
\begin{equation*}
\begin{split}
\Delta_2=\|Y_2\alpha_2\|^2
=\alpha_2^t \cdot Y_2^t Y_2\cdot \alpha_2
&=\alpha_2^t\cdot \ell_2 K_{d_2}\cdot \alpha_2\\
&=\ell_2( \|\alpha_2\|^2+{\rm wt}(\alpha_2)^2),
\end{split}
\end{equation*}
which is the second equality.
The first and third equalities can be proved similarly.
\end{proof}

Lemma \ref{lemma:caldelta} allows us to represent $\|\psi\|^2$ as an
explicit function of $\alpha_1$, $\alpha_2$ and $\alpha_3$. Let
\begin{equation}
\label{equ:S_12}
\begin{split}
S_1&=\|\alpha_2\|^2+{\rm wt}(\alpha_2)^2 \quad {\rm and} \quad \\
S_2&=\|\alpha_3\|^2+{\rm wt}(\alpha_3)^2.
\end{split}
\end{equation}
\begin{lemma}
\label{lemma:representation}
We have that
$
\|\psi\|^2=\lambda_4 |X|+\ell(\lambda_1-\lambda_4)\alpha_{11}^2
+\ell_2(\lambda_2-\lambda_4)S_1+\ell_1(\lambda_3-\lambda_4)S_2
$.
\end{lemma}
\begin{proof}
From Lemma \ref{lemma:eigenvector}, we have that $|X|=\|\chi_{\mbox{\tiny\itshape X}} \|^2=
\Delta_1+\Delta_2+\Delta_3+\Delta_4$. It follows that $\Delta_4=|X|-
\Delta_1-\Delta_2-\Delta_3$.  Along with
(\ref{eqn:normpsi}), (\ref{eqn:defDs}), (\ref{eqn:caldelta}) and (\ref{equ:S_12}), this
implies the expected equality.
\end{proof}

Although Lemma \ref{lemma:representation} gives us a representation of $\|\psi\|^2$
in terms of $|X|$, $\alpha_{11}$, $S_1$ and $S_2$, it can be more explicit if we know
how the quantities $\alpha_{11}$, $S_1$ and $S_2$ are
connected to $X$.
Note that
 $\alpha=Y^{-1}  \chi_{\mbox{\tiny\itshape X}} $ according to  (\ref{eqn:chi}).
Let $Z_1=\ell^{-1}\cdot \textbf{1}_\ell$, $Z_2=\ell^{-1}\cdot L_{d_2} \otimes \textbf{1}_{\ell_2}$ and $
Z_3=\ell^{-1}\cdot \textbf{1}_{\ell_1}\otimes L_{d_3}$. Then
\begin{equation}
\label{equ:alpha_s}
\alpha_s=Z_s  \chi_{\mbox{\tiny\itshape X}}
\end{equation}
for every  $s\in \{1,2,3\}$, by Lemma \ref{lemma:eigenvector}.
As an immediate consequence, we then have that
\begin{equation}
\label{eqn:alpha1}
\alpha_{11}=\alpha_1=Z_1 \chi_{\mbox{\tiny\itshape X}} =\ell^{-1}|X|.
\end{equation}
On the other hand, recall that
 $\mathbb{P}_{n,p}$, $\mathbb{P}_{n,q}$ and $\mathbb{P}_{n,m}$
 have been made unique as ordered sets
in Section \ref{sec:pg}. For every $h\in [\ell]$,
 there exists $(i,j)\in [\ell_1]\times [\ell_2]$ such that ${\bf w}_h=\pi({\bf u}_i, {\bf v}_j)$.
Let $\sigma: \mathbb{P}_{n,m}\rightarrow [\ell_1] $ be the mapping defined by
\begin{equation}
\label{eqn:sigma}
\sigma({\bf w}_h)=\left\lfloor\frac{h-1}{\ell_2}\right\rfloor+1
\end{equation}
and let  $\tau:\mathbb{P}_{n,m}\rightarrow [\ell_2]$ be the mapping defined by
\begin{equation}
\label{eqn:tau}
\tau({\bf w}_h)=h-(\sigma({\bf w}_h)-1)\ell_2.
\end{equation}

\begin{lemma}
\label{lem:detw}
We have that ${\bf w}_h=\pi({\bf u}_{\sigma({\bf w}_h)}, {\bf v}_{\tau({\bf w}_h)})$
for every $h\in [\ell]$.
\end{lemma}
\begin{proof}
Suppose that ${\bf w}_h=\pi({\bf u}_i, {\bf v}_j)$ for $(i,j)\in [\ell_1]\times [\ell_2]$.
  Then the representation of $\mathbb{P}_{n,m}$ in Section \ref{sec:pg} shows that
$h=(i-1)\ell_2+j$. It is easy to see that $i=\sigma({\bf w}_h)$ and $j=\tau({\bf w}_h)$.
\end{proof}

For every $i\in [\ell_1]$ and $j\in [\ell_2]$,  let
$
\sigma^{-1}(i)$ be the preimage of $i$ under $\sigma$ and let
$\tau^{-1}(j)$ be the  preimage of $j$ under $\tau$.
Let ${\bf a}\in \mathbb{R}^{\ell_1}$ and ${\bf b}\in\mathbb{R}^{\ell_2}$ be  two real vectors
defined by
\begin{equation}
\label{eqn:defab}
\begin{split}
{\bf a}(i)=|\sigma^{-1}(i)\cap X| \quad {\rm and} \quad
{\bf b}(j)=|\tau^{-1}(j)\cap X|,
\end{split}
\end{equation}
where  $i\in [\ell_1]$ and $j\in [\ell_2]$.
Then we clearly  have that ${\rm wt}({\bf a})=
{\rm wt}({\bf b})=|X|$.
\begin{lemma}
\label{lem:S_12}
We have that
$S_1 = \ell^{-2}\ell_1 (\ell_1 \|{\bf a}\|^2-|X|^2)$ and $
S_2 = \ell^{-2}\ell_2 (\ell_2 \|{\bf b}\|^2-|X|^2)$.
\end{lemma}
\begin{proof}
For every $i\in [d_2]$, the $i$th row  of $Z_2$ is
\begin{equation*}
\begin{split}
Z_2[i]&=\ell^{-1}
\begin{pmatrix}
-\textbf{1}_{i-1} & \ell_1-1 & -\textbf{1}_{\ell_1-i}
\end{pmatrix} \otimes \textbf{1}_{\ell_2}\\
&=
\ell^{-1}\ell_1
\begin{pmatrix}
\textbf{0}_{i-1}  &  1 & \textbf{0}_{\ell_1-i}
\end{pmatrix}\otimes \textbf{1}_{\ell_2}-\ell^{-1} \textbf{1}_{\ell}.
\end{split}
\end{equation*}
Let $T = \ell^{-1}\ell_1
\begin{pmatrix}
\textbf{0}_{i-1}  &  1 & \textbf{0}_{\ell_1-i}
\end{pmatrix}\otimes \textbf{1}_{\ell_2}$, so that
\begin{equation*}
Z_2[i] = T - \ell^{-1} \textbf{1}_{\ell}.
\end{equation*}
The nonzero components of  $T$ are labeled by
$\sigma^{-1}(i)$. It follows that
 $T \cdot  \chi_{\mbox{\tiny\itshape X}} =\ell^{-1}\ell_1 {\bf a}(i)$ and therefore
$$\alpha_{2i}= Z_2[i]\cdot \chi_{\mbox{\tiny\itshape X}} =T\cdot \chi_{\mbox{\tiny\itshape X}} -\ell^{-1} \textbf{1}_{\ell}\cdot \chi_{\mbox{\tiny\itshape X}}
=\ell^{-1}(\ell_1\cdot {\bf a}(i)-|X|).$$
Note that ${\rm wt}({\bf a})={\bf a}(1)+\cdots+{\bf a}(\ell_1)=|X|$ and $d_2=\ell_1-1$.
Due to (\ref{equ:S_12}), we have that
\begin{equation*}
\begin{split}
S_1=\|\alpha_2\|^2+{\rm wt}(\alpha_2)^2
&=\sum_{i=1}^{d_2} \alpha_{2i}^2+ \Big(\sum_{i=1}^{d_2} \alpha_{2i}\Big)^2\\
&=\ell^{-2}\ell_1(\ell_1\cdot \|{\bf a}\|^2-|X|^2),
\end{split}
\end{equation*}
which is the first equality.

For every $j\in[d_3]$, the  $j$th row  of $Z_3$ is
\begin{equation*}
\begin{split}
Z_3[j]&=\ell^{-1}
\textbf{1}_{\ell_1}\otimes
\begin{pmatrix}
-\textbf{1}_{j-1} & \ell_2-1 & -\textbf{1}_{\ell_2-j}
\end{pmatrix}\\
&=
\ell^{-1}\ell_2 \textbf{1}_{\ell_1}\otimes
\begin{pmatrix}
\textbf{0}_{j-1} & 1 & \textbf{0}_{\ell_2-j}
\end{pmatrix}-\ell^{-1} \textbf{1}_{\ell}.
\end{split}
\end{equation*}
Let $T' = \ell^{-1}\ell_2 \textbf{1}_{\ell_1}\otimes
\begin{pmatrix}
\textbf{0}_{j-1} & 1 & \textbf{0}_{\ell_2-j}
\end{pmatrix}$, so that
\begin{equation*}
Z_3[j] = T' - \ell^{-1} \textbf{1}_{\ell}.
\end{equation*}
The nonzero components of $T'$ are labeled by
$\tau^{-1}(j)$.
It follows that
 $T' \cdot  \chi_{\mbox{\tiny\itshape X}} =\ell^{-1}\ell_2 {\bf b}(j)$ and therefore
$$\alpha_{3j}= Z_3[j]\cdot \chi_{\mbox{\tiny\itshape X}}  =T' \cdot \chi_{\mbox{\tiny\itshape X}} -\ell^{-1} \textbf{1}_{\ell}\cdot \chi_{\mbox{\tiny\itshape X}}
=\ell^{-1}(\ell_2 {\bf b}(j)-|X|).$$
Note that ${\rm wt}({\bf b})={\bf b}(1)+\cdots+{\bf b}(\ell_2)=|X|$ and $d_3=\ell_2-1$.
Due to (\ref{equ:S_12}), we have that
\begin{equation*}
\begin{split}
S_2=\|\alpha_3\|^2+{\rm wt}(\alpha_3)^2
&=\sum_{i=1}^{d_3} \alpha_{3i}^2+ \Big(\sum_{i=1}^{d_3} \alpha_{3i}\Big)^2\\
&=\ell^{-2}\ell_2(\ell_2\cdot \|{\bf b}\|^2-|X|^2),
\end{split}
\end{equation*}
which is the second equality.
\end{proof}

Lemmas \ref{lemma:representation} and \ref{lem:S_12}, together with (\ref{eqn:alpha1}),
result in an explicit representation of $\|\psi\|^2$ in terms of $X$:
\begin{equation}
\label{equ:normpsi}
\begin{split}
\|\psi\|^2
=\lambda_4|X|&+\ell^{-1}(\lambda_1-\lambda_4)|X|^2\\
&+(\lambda_2-\lambda_4) \ell^{-1}(\ell_1  \|{\bf a}\|^2-|X|^2)
\\
&+(\lambda_3-\lambda_4)\ell^{-1} (\ell_2 \|{\bf b}\|^2-|X|^2).
\end{split}
\end{equation}
For simplicity, we denote by $F({\bf a}, {\bf b})$ the
right hand side of Equation (\ref{equ:normpsi}).
We aim to deduce an upper bound for $F({\bf a}, {\bf b})$
 in terms
of $|X|$. Clearly, this also provides an upper bound for $\|\psi\|^2$ and
is crucial
 for
establishing that ${\bf G}_{n,m}$ has some kind of expanding property.
Let
\begin{equation}
\begin{split}
\kappa_p=\lfloor 4p^{0.5n}+2 \rfloor
\quad {\rm and} \quad
\kappa_q=\lfloor 4q^{0.5n}+2\rfloor.
\end{split}
\end{equation}
Dvir {\it et al.}~\cite{DGY10} showed that $k(p,n)\leq \kappa_p$ and $k(q,n)\leq \kappa_q$.
Let ${\cal U}, {\cal V}\subseteq \mathbb{P}_{n,m}$ form a matching family. From now on, we
suppose that $X\subseteq {\cal U}$ and, furthermore, that its cardinality
$|X|=x\leq \min\{\kappa_q \ell_1, \kappa_p \ell_2\}$ is fixed. We remark that this assumption
does not affect our proof adversely (see
Theorem \ref{thm:upper_bound_m_infty}).

\begin{lemma}\label{lem:bound_aibj}
Let  ${\bf a}, {\bf b}$ be the real vectors defined by {\rm (\ref{eqn:defab})}.
 Then we have that ${\bf a}(i)\leq \kappa_q$ for every $i\in[\ell_1]$ and
  ${\bf b}(j)\leq \kappa_p$ for every
$j\in [\ell_2]$.
\end{lemma}

\begin{proof}
Suppose that  ${\bf a}(i)>\kappa_q$ for some $i\in [\ell_1]$.
Let $\mathcal{U}^\prime= \sigma^{-1}(i)\cap X\triangleq \{{\bf u}^\prime_{s}:
s\in [{\bf a}(i)]\}\subseteq {\cal U}$.  Then, by the definition
of matching families,  there is a subset  of $\mathcal{V}$, say
$\mathcal{V}^\prime=\{{\bf v}^\prime_{s}: s\in [{\bf a}(i)]\}$, such that
$\mathcal{U}^\prime$ and $\mathcal{V}^\prime$ form
 a matching family.
It follows that
\begin{itemize}
\item $ \langle {\bf u}^\prime_{s},  {\bf v}^\prime_{s}\rangle
\equiv 0\bmod m$ for every $s\in [{\bf a}(i)]$;
\item $\langle  {\bf u}^\prime_{s},  {\bf v}^\prime_{t}\rangle
 \not\equiv 0 \bmod m$ whenever $s,t\in [ {\bf a}(i)]$ and  $s\neq t$.
\end{itemize}
On the one hand, we immediately have that
\begin{itemize}
\item $ \langle {\bf u}^\prime_{s},  {\bf v}^\prime_{s}\rangle
\equiv 0\bmod q$ for every $s\in [{\bf a}(i)]$.
\end{itemize}
On the other hand, Lemma \ref{lem:detw} shows that any two elements in $\mathbb{P}_{n,m}$ are equivalent to each other
as elements of $\mathbb{Z}_p^n$ as long as they have the same image under $\sigma$.
Therefore,  ${\bf u}^\prime_s\sim {\bf u}^\prime_t$ as elements  of $\mathbb{Z}_p^n$
for any  $s,t\in [ {\bf a}(i)]$. It follows that
$\langle {\bf u}^\prime_{s}, {\bf v}^\prime_{t} \rangle \equiv
\langle {\bf u}^\prime_{t}, {\bf v}^\prime_{t} \rangle \equiv 0\bmod p$. Recall that
 $\langle  {\bf u}^\prime_{s},  {\bf v}^\prime_{t}\rangle
 \not\equiv 0 \bmod m$ whenever $s\neq t$. It follows that
\begin{itemize}
\item $\langle  {\bf u}^\prime_{s},  {\bf v}^\prime_{t}\rangle
 \not\equiv 0 \bmod q$ whenever $s,t\in [ {\bf a}(i)]$ and  $s\neq t$.
\end{itemize}
Therefore, $\mathcal{U}^{\prime}$ and $\mathcal{V}^{\prime}$ form a matching family
in $\mathbb{Z}_q^{n}$
 of size ${\bf a}(i)>\kappa_q$, which contradicts Dvir {\it et al.} \cite{DGY10}.
Hence, we must have that ${\bf a}(i)\leq \kappa_q$ for every
$i\in [\ell_1]$.

Similarly,  we must have that
${\bf b}(j)\leq \kappa_p$ for every  $j\in [\ell_2]$.
\end{proof}

Lemma \ref{lem:bound_aibj} shows that
the components of ${\bf a}$ and ${\bf b}$ cannot be too large when $X\subseteq \mathcal{U}$.
In fact, several conditions must be satisfied by the real vectors ${\bf a}$ and ${\bf b}$, which
can be summarized as follows:
\begin{itemize}
\item $0\leq {\bf a}(i)\leq \kappa_q$ for every $i\in [\ell_1]$, due to Lemma \ref{lem:bound_aibj}
and Equation (\ref{eqn:defab});
\item $0\leq {\bf b}(j)\leq \kappa_p$ for every $j\in [\ell_2]$, due to
 Lemma \ref{lem:bound_aibj}
and Equation (\ref{eqn:defab});
\item  ${\rm wt}({\bf a})={\rm wt}({\bf b})=|X|=x$, due to Equation (\ref{eqn:defab}).
\end{itemize}
Clearly, when $x$ is fixed,  the problem of establishing  an upper bound for $F({\bf a}, {\bf b})$ can be reduced to that of
determining the maximum value of $F({\bf a}, {\bf b})$ subject to the  conditions
above.
Let
 \begin{equation}
\begin{split}
\mu_q&=\left\lfloor \frac{x}{\kappa_q}\right \rfloor, \quad
\nu_q=x-\kappa_q\mu_q, \quad\\
{\bf a}^*&=\begin{pmatrix}  \kappa_q\cdot \textbf{1}_{\mu_q} & \nu_q &
\textbf{0}_{\ell_1-1-\mu_q}
\end{pmatrix},\\
\mu_p&=\left\lfloor \frac{x}{\kappa_p}\right \rfloor, \quad
\nu_p=x-\kappa_p\mu_p, \quad\\
{\bf b}^* &= \begin{pmatrix}  \kappa_p\cdot \textbf{1}_{ \mu_p} &
\nu_p & \textbf{0}_{\ell_2-1-\mu_p} \end{pmatrix}.
\end{split}
\end{equation}
We show below that $F({\bf a}^*, {\bf b}^*)$  is the maximum value of
$F({\bf a}, {\bf b})$ subject to the conditions above.

\begin{lemma}
\label{lem:comp}
Let $a,b,c,d\in \mathbb{N}$ be such that $a\geq b$,  $c\geq d$, and $a+b=c+d$. If $a\geq c$, then
$a^2+b^2\geq c^2+d^2$.
\end{lemma}
\begin{proof}
Clearly, we have that $a^2+b^2-c^2-d^2=(a-c)(a+c)+(b-d)(b+d)=(a-c)(a+c)-(a-c)
(b+d)=(a-c)(a+c-b-d)\geq 0$, where the second equality follows from $a+b=c+d$ and the last
inequality follows from $a\geq b$, $c\geq d$ and $a\geq c$.
\end{proof}

\begin{lemma}
\label{lem:boundS}
We have that $\|\psi\|^2=F({\bf a}, {\bf b})\leq F({\bf a}^*, {\bf b}^*)$.
\end{lemma}

\begin{proof}
First, we note that the  vectors ${\bf a}^*$ and ${\bf b}^*$ satisfy the three conditions above.
In order to show that $F({\bf a}, {\bf b})\leq F({\bf a}^*, {\bf b}^*)$, in view of  (\ref{equ:normpsi}),
 it suffices to show that
$\|{\bf a}\|^2\leq \|{\bf a}^*\|^2$ and
$\|{\bf b}\|^2\leq \|{\bf b}^*\|^2$.
We only show the first inequality. The second one can be proved similarly.

Without loss of generality, we can suppose that
${\bf a}(1)\geq {\bf a}(2)\geq \cdots \geq {\bf a}(\ell_1)$.
From Lemma \ref{lem:bound_aibj}, we have that ${\bf a}(i)\leq \kappa_q$ for every
$i\in [\ell_1]$. We claim that the algorithm in Figure 2  takes as input the original vector
${\bf a}_0={\bf a}$ and produces a sequence of vectors, say ${\bf a}_0,{\bf a}_1,\ldots, {\bf a}_h$,
such that
\begin{itemize}
\item  $\kappa_q\geq {\bf a}_s(1)\geq {\bf a}_s(2)
\geq \cdots \geq {\bf a}_s(\ell_1)\geq 0$ and
${\rm wt}({\bf a}_s)=x$ for every $s\in \{0,1,\ldots, h\}$; and
\item ${\bf a}_0={\bf a}$, ${\bf a}_h={\bf a}^*$ and $\|{\bf a}_s\|^2\leq \|{\bf a}_{s+1}\|^2$
for every $s\in \{0,1,\ldots, h-1\}$.
\end{itemize}
Clearly, if the algorithm does have the above functionality, then we must have that $\|{\bf a}\|^2
\leq \|{\bf a}^*\|^2$.

Consider the algorithm in Figure 2. In order to get the expected
sequence, i.e.,  ${\bf a}_0,{\bf a}_1,\ldots, {\bf a}_h$, it will be run with an initial
 input $
{\bf c}={\bf a}_0={\bf a}$. In every iteration, the algorithm
 outputs a ${\bf c}^\prime$ and then checks whether   ${\bf c}^\prime={\bf a}^*$.
 It halts once  the equality holds.
\begin{figure*}[t]
\begin{center}
Figure 2: An Algorithm

\vspace{0.25cm}

\begin{boxedminipage}{18cm}
while ${\bf c}\neq {\bf a}^*$ do
\begin{itemize}
\item set $i_0=\min \{i\in [\ell_1]: {\bf c}(i)\neq {\bf a}^*(i)\}$;
\item set $a=\left\{
\begin{aligned}
\min\{\kappa_q, {\bf c}(i_0)+{\bf c}(i_0+1)\}, &  {\rm ~if~} i_0\leq \mu_q, \\
\min\{\nu_q, {\bf c}(i_0)+{\bf c}(i_0+1)\}, &  {\rm ~if~} i_0= \mu_q+1;
\end{aligned}
\right.$
\item set  $b={\bf c}(i_0)+{\bf c}(i_0+1)-a$;
\item set  $j_0=\min \big( \{j: j\in \{i_0+2,\ldots, \ell_1\}\wedge
{\bf c}(j)\leq b \}\cup \{0\}\big)$;
\begin{itemize}
\item if $j_0=i_0+2$, set ${\bf c}^\prime=({\bf c}(1),\ldots,{\bf c}(i_0-1),a, b,
{\bf c}(i_0+2), \ldots, {\bf c}(\ell_1))$;
\item if $j_0=0$ or $\ell_1$, set ${\bf c}^\prime=({\bf c}(1),\ldots,{\bf c}(i_0-1),a,
{\bf c}(i_0+2), \ldots, {\bf c}(\ell_1), b)$;
\item otherwise, set
${\bf c}^\prime=({\bf c}(1),\ldots,{\bf c}(i_0-1),a, {\bf c}(i_0+2), \ldots,
{\bf c}(j_0-1), b, {\bf c}(j_0),\ldots, {\bf c}(\ell_1))$;
\end{itemize}
\item set ${\bf c}={\bf c}^\prime$.
\end{itemize}
\end{boxedminipage}
\end{center}
\label{fig:algo}
\end{figure*}

We must show that this algorithm does halt in a finite number of steps and achieves
the promised functionality.
The algorithm starts with ${\bf c}$ and checks whether ${\bf c}={\bf a}^*$. If the
equality holds, it halts. Otherwise, it constructs a new vector ${\bf c}^\prime$
such that ${\rm wt}({\bf c}^\prime)=x$,
$\kappa_q\geq {\bf c}^\prime(1)\geq {\bf c}^\prime(2)\geq \cdots
\geq {\bf c}^\prime(\ell_1)\geq 0$
and $\|{\bf c}\|^2\leq \|{\bf c}^\prime\|^2$.
More concretely,  the algorithm finds the first coordinate (say $i_0\in [\ell_1]$)
where ${\bf c}$ and ${\bf a}^*$ differ.
Clearly, we have that $i_0\leq \mu_q+1$,
 ${\bf c}(i)={\bf a}^*(i)$ for every $i\in \{1,2,\cdots, i_0-1\}$ and
 ${\bf c}(i_0)<{\bf a}^*(i_0)$.
 Next, the algorithm
effects a carry from ${\bf c}(i_0+1)$ to ${\bf c}(i_0)$. This is
done by setting ${\bf c}^\prime(i_0)=a$. Finally, the algorithm must
determine ${\bf c}^\prime(i)$ for every $i\in \{i_0+1,\ldots,\ell_1\}$.
This is done by rearranging  the $\ell_1-i_0$ numbers
$b, {\bf c}(i_0+2), \ldots, {\bf c}(\ell_1)$ such that they are in descending
order. By the description above, it is clear that
\begin{itemize}
\item $ {\rm wt}({\bf c}^\prime)=\sum_{i=1}^{i_0-1} {\bf c}(i)+a+b+\sum_{i=i_0+2}^{\ell_1}{\bf c}(i)
={\rm wt}({\bf c})=x$;
\item  ${\bf c}^\prime(1)={\bf c}^\prime(2)=\cdots={\bf c}^\prime(i_0-1)=\kappa_q\geq
{\bf c}^\prime(i_0)=a>{\bf c}(i_0)\geq {\bf c}^\prime(i_0+1)\geq \cdots
\geq  {\bf c}^\prime(\ell_1)$;
\item $0\leq {\bf c}^\prime(i)\leq \kappa_q$ for every $i\in [\ell_1]$;
\item $\|{\bf c}^\prime\|^2-\|{\bf c}\|^2=
a^2+b^2-{\bf c}(i_0)^2-{\bf c}(i_0+1)^2\geq 0$, due to Lemma \ref{lem:comp}.
\end{itemize}

In each iteration,  either $i_0$ becomes greater than it was
in the previous iteration or $i_0$ does not change but
the new ${\bf c}^\prime(i_0)$ obtained is strictly greater than ${\bf c}(i_0)$.
However, since ${\bf c}^\prime(i_0)$ must be bounded by $\kappa_q$, in the latter case,
${\bf c}^\prime(i_0)$ will eventually become ${\bf a}^*(i_0)$ in a finite number of
iterations. Then, in the next iteration, $i_0$ will be increased by at least 1.
Therefore, we can get a sequence ${\bf a}_0={\bf a}, {\bf a}_1,
\cdots, {\bf a}_h={\bf a}^*$, where $h$ is the number of iterations.
\end{proof}

Lemma \ref{lem:boundS} shows that $F({\bf a}^*,{\bf b}^*)$ is a valid
upper bound for $\|\psi\|^2$. This bound is nice because both ${\bf a}^*$ and
${\bf b}^*$ depend only on $x$, which will facilitate our analysis.
More precisely, we have that
\begin{equation}
\label{equation:upper-bound-norm}
\|\psi\|^2\leq \ell^{-1}\lambda_1 x^2+\Delta,
\end{equation}
where $\Delta=\lambda_4x-\ell^{-1}\lambda_4x^2
+(\lambda_2-\lambda_4) \ell^{-1}(\ell_1  \|{\bf a}^*\|^2-x^2)+
(\lambda_3-\lambda_4)\ell^{-1} (\ell_2 \|{\bf b}^*\|^2-x^2)$.

We proceed to   develop an  explicit  lower bound for $\|\psi\|^2$ in terms of $x$ and $|N(X)|$.
Recall that
the components  of $\psi$ are labeled by all the hyperplanes. It is easy to see that
\begin{equation}
\psi({\bf v})=|N({\bf v})\cap X|
\end{equation}
is the number of neighbors of ${\bf v}$ in $X$ for every  ${\bf v}\in \mathbb{H}_{n,m}$.
 Hence, $\psi({\bf v})=0$ whenever
 ${\bf v}\notin N(X)$. It follows that
\begin{equation}
\label{equation:weight}
\begin{split}
\sum_{{\bf v}\in \mathbb{H}_{n,m}} \psi({\bf v}) &=\sum_{{\bf v}\in N(X)}\psi({\bf v})\\
&=\sum_{{\bf u}\in X}
|N({\bf u})|=x\cdot  \theta_{n-1,m},
\end{split}
\end{equation}
where the last equality follows from Chee and Ling \cite{CL93}.
It follows from the Cauchy-Schwarz inequality that
\begin{equation}\label{equation:lower-bound-norm}
\begin{split}
\|\psi\|^2 &=\sum_{{\bf v}\in N(X)} \psi({\bf v})^2
\geq
\frac{1}{|N(X)|}\bigg(\sum_{{\bf v}\in N(X)}\psi({\bf v})\bigg)^2\\
&=\frac{x^2\theta_{n-1,m}^2}{|N(X)|}=\frac{\lambda_1 x^2}{|N(X)|},
\end{split}
\end{equation}
where the second equality follows from Equation (\ref{equation:weight}) and the last equality
follows from Lemma \ref{lemma:eigenvalue}.

Both the upper bound (see Equation (\ref{equation:upper-bound-norm})) and the
lower bound (see Equation (\ref{equation:lower-bound-norm})) for
$\|\psi\|^2$ involve only $x$ and $|N(X)|$. Together, they demonstrate that
the projective graph ${\bf G}_{n,m}$ has some kind of expanding property.

\begin{theorem}\label{thm:bound_NX}
{\em (Expanding Property)}
Let $\mathcal{U}, \mathcal{V}\subseteq \mathbb{P}_{n,m}$ form a  matching family
and let $X\subseteq \mathcal{U}$ be of cardinality $x \leq \min\{\kappa_q \ell_1, \kappa_p \ell_2\}$. Then
we have that
$
|N(X)|\geq \lambda_1x^2/(\ell^{-1}\lambda_1x^2+\Delta).
$
\end{theorem}

\subsection{On the Largest Matching Family in $\mathbb{P}_{n,m}$}

In this section, we deduce an upper bound on the largest matching family in
$\mathbb{P}_{n,m}$.
As in \cite{DGY10}, our analysis depends on the
unique neighbor property  defined
in Definition \ref{def:unp} and the fact that the projective
graph ${\bf G}_{n,m}$ has some kind of expanding property (see Theorem \ref{thm:bound_NX}).
\begin{definition}\label{def:unp}
{\em (Unique Neighbor Property)}
We say that $\mathcal{U}\subseteq \mathbb{P}_{n,m}$ satisfies
the {\em unique neighbor property}   if, for every ${\bf u}\in \mathcal{U}$, there is
a ${\bf v}\in N({\bf u})$ such that ${\bf v}$ is not  adjacent to any ${\bf w}\in \mathcal{U}\setminus \{{\bf u}\}$.
\end{definition}

As noted by Dvir {\it et al.}~\cite{DGY10},  there is a set $\mathcal{U}\subseteq \mathbb{P}_{n,p}$ of cardinality $k$
that satisfies the  unique neighbor property in ${\bf G}_{n,p}$ if and only if there is a matching family in
$\mathbb{Z}_p^{n}$ of size $k$.
As an analogue, the following lemma is true for ${\bf G}_{n,m}$.
\begin{lemma}
\label{lem:unpm}
A set   $\mathcal{U}\subseteq \mathbb{P}_{n,m}$
 satisfies the  unique neighbor property if and only if there is
a  $\mathcal{V}\subseteq \mathbb{H}_{n,m}$  such that
$\mathcal{U}$ and $\mathcal{V}$ form a  matching family.
\end{lemma}
\begin{proof}
Suppose that $\mathcal{U}=\{{\bf u}_1,\ldots, {\bf u}_k\}$. If it satisfies the unique neighbor property in
${\bf G}_{n,m}$, then, for every $i\in[k]$, there is a ${\bf v}_i\in N({\bf u}_i)$ such that
 ${\bf v}_i\notin N({\bf u}_j)$ for every $j\in [k]\setminus \{i\}$. Equivalently, we have that
$\langle {\bf u}_i, {\bf v}_i \rangle=0$  and $\langle  {\bf u}_j, {\bf v}_i\rangle\ne 0$
for every $j\in [k]\setminus \{i\}$.
Let $\mathcal{V}=\{{\bf v}_1,\ldots, {\bf v}_k\}$. Then $\mathcal{U}$ and $\mathcal{V}$
form a matching family.

Conversely, let
 $\mathcal{V}=\{{\bf v}_1,\ldots, {\bf v}_k\}\subseteq \mathbb{H}_{n,m}$ be such that
  $\mathcal{U}$ and $\mathcal{V}$ form a matching family. For every $i\in [k]$, we have that $\langle {\bf u}_i, {\bf v}_i \rangle=0$  and
 $\langle {\bf u}_j, {\bf v}_i \rangle\ne 0$ whenever $j\in [k] $ and $j\ne i$.
Equivalently,
${\bf v}_i\in N({\bf u}_i)$  and $ {\bf v}_i\notin N({\bf u}_j)$ when
$j\in [k]\setminus \{i\}$.
 Hence,  $\mathcal{U}$ satisfies the unique neighbor property.
\end{proof}

\begin{theorem}\label{thm:general_bound}
Let $\mathcal{U},  \mathcal{V}\subseteq \mathbb{P}_{n,m}$ form a  matching family
and let $X\subseteq \mathcal{U}$ be of cardinality $x$. Then we have that
$
|\mathcal{U}|\leq x+\ell\Delta/(\ell^{-1}\lambda_1x^2+\Delta).
$
\end{theorem}

\begin{proof}
By Lemma \ref{lem:unpm},  $\mathcal{U}$  satisfies the unique neighbor property in
${\bf G}_{n,m}$.  Hence, every element in $\mathcal{U}\setminus X$ must have a unique neighbor in
$\mathbb{H}_{n,m}\setminus N(X)$. It follows that
$|\mathcal{U}\setminus X|\leq |\mathbb{H}_{n,m}\setminus N(X)|=\ell-|N(X)|$, which implies that
$|\mathcal{U}|\leq |X|+\ell-|N(X)|$. Along with Theorem \ref{thm:bound_NX},
the inequality desired follows.
\end{proof}

The following theorem gives an explicit upper bound for the largest matching family
in $\mathbb{P}_{n,m}$.
\begin{theorem}\label{thm:upper_bound_m_infty}
Let $\mathcal{U}, \mathcal{V}\subseteq \mathbb{P}_{n,m}$ form a  matching family.
Then
$
 |\mathcal{U}|\leq (8+\epsilon) m^{0.625n+0.125}
$
  for any constant $\epsilon>0$ as
 $p \rightarrow \infty$ and $p/q\rightarrow 1$,  where  $n$ is a constant.
\end{theorem}
\begin{proof}
Suppose that $|\mathcal{U}|> (8+\epsilon) m^{0.625n+0.125}$. Then
we can  take a set of points $X\subseteq \mathcal{U}$  of cardinality
$x=\lfloor \ell^{0.625} \rfloor\leq \min\{\kappa_q \ell_1, \kappa_p \ell_2\}$.
By Theorem \ref{thm:general_bound},
we have that
$|\mathcal{U}|\leq x+\ell\Delta/(\ell^{-1}\lambda_1x^2+\Delta)\approx
 8m^{0.625n+0.125}$
when  $p\rightarrow \infty$ and $ p/q\rightarrow 1$, with
$n$ a constant. This is   a contradiction.
\end{proof}

\subsection{On the Largest Matching Family in $\mathbb{Z}_{m}^n$}

Let $\mathcal{U}=\{{\bf u}_1,\ldots, {\bf u}_k\}, \mathcal{V}=\{{\bf v}_1,\ldots, {\bf v}_k\}$
be a  matching family of size $k=k(m,n)$  in $ \mathbb{Z}_m^{n}$.
In order to establish the final upper bound for $k(m,n)$, we have to classify the
pairs $\{({\bf u}_i, {\bf v}_i): i\in [k]\}$ into types and establish an upper bound for each type.

\begin{definition}\label{def:type}
{\em (Type of Pairs)}
For every $i\in [k]$, the pair  $({\bf u}_i, {\bf v}_i)$ is said to be of {\em type $(s,t)$}
 if $\gcd({\bf u}_i(1),\ldots, {\bf u}_i(n), m)=s$ and
$\gcd({\bf v}_i(1),\ldots, {\bf v}_i(n), m)=t$, where $s, t$ are positive divisors of $m$.
\end{definition}

Let $s,t\in \{1,p,q,m\}$. We define
$\Omega_{s,t}$ to be the set of pairs $({\bf u}_i, {\bf v}_i)$ of type $(s,t)$ and
$
N_{s,t}=|\Omega_{s,t}|.
$
Clearly, the elements of the set $\{({\bf u}_i, {\bf v}_i): i\in [k]\}$
fall into 16 different classes as $s$ and $t$ vary.
\begin{lemma}
\label{lem:stdm}
If $m|st$, then $N_{s,t}\leq 1$.
\end{lemma}

\begin{proof}
Suppose that $N_{s,t}>1$. Then we can take two pairs, say  $({\bf u}_1,{\bf v}_1)$ and $({\bf u}_2, {\bf v}_2)$, from $\Omega_{s,t}$.
Clearly, we have that $\langle {\bf u}_1, {\bf v}_2\rangle=\langle
 {\bf u}_2, {\bf v}_1\rangle=0$, which is a contradiction.
\end{proof}

Lemma \ref{lem:stdm} covers nine of the 16 classes. More precisely, we have that $N_{s,t}\leq 1$ when
$(s,t)\in \{(p,q),~ (q,p),~ (m,1),~ (m,p),~ (m,q),~ (m,m),~ (1,m),~
 (p,m),\\~ (q,m)\}$.
\begin{lemma}\label{lem:p_family}
If   $(s,t)\in \{(1,p),(p,1),  (p,p)\}$,
then $N_{s,t}\leq \kappa_q$.
\end{lemma}

\begin{proof}
We prove for $(s,t)=(1,p)$. The other cases can be treated similarly.
Without loss of generality, we can suppose that $\{({\bf u}_1, {\bf v}_1), \ldots,
({\bf u}_c, {\bf v}_c)\}$ are the pairs of type $(s,t)$, where $c=N_{1,p}$.
Let
\begin{itemize}
\item
${\bf u}_i^\prime=({\bf u}_i(1) \bmod q, \ldots, {\bf u}_i(n)\bmod q)$ and
\item
${\bf v}_i^\prime=({\bf v}_i(1)/p \bmod q, \ldots, {\bf v}_i(n)/p \bmod q)$
\end{itemize}
for every $i\in [c]$.
Then  $\mathcal{U}^\prime=\{{\bf u}_1^\prime, \ldots, {\bf u}_c^\prime\}$ and
$\mathcal{V}^\prime=\{{\bf v}_1^\prime, \ldots,
{\bf v}_c^\prime\}$ form  a matching family of size $c$ in
$\mathbb{Z}_q^n$. This implies that $N_{s,t}=c\leq \kappa_q$, using results of Dvir {\it et al.} in \cite{DGY10}.
\end{proof}

Similarly, we have

\begin{lemma}\label{lem:q_family}
 If   $(s,t)\in \{(1,q),(q,1),  (q,q)\}$,
then $N_{s,t}\leq \kappa_p$.
\end{lemma}

Finally, we have the main result of this paper.

\begin{theorem}\label{thm:main}
Let $n$ be a constant and let $m=pq$ for two distinct primes $p$ and $q$.
Then we have that
$k(m,n)\leq O(m^{0.625n+0.125})$
when $p\rightarrow \infty$ and $p/q\rightarrow 1$.
\end{theorem}

\begin{proof}
Theorem  \ref{thm:upper_bound_m_infty} gives an upper bound for $N_{1,1}$.
By Theorem  \ref{thm:upper_bound_m_infty}, Lemmas \ref{lem:stdm}, \ref{lem:p_family} and  \ref{lem:q_family},
$k(m,n)=k=\sum_{s|m, t|m}N_{s,t}\leq 9+ 3\kappa_p+3\kappa_q+O\big(m^{0.625n+0.125}\big)$,
which is asymptotically bounded by $O\big(m^{0.625n+0.125}\big)$ when $p\rightarrow \infty$ and $p/q
\rightarrow 1$.
\end{proof}

\section{Concluding Remarks}\label{sec:con}

It would be attractive if the method in this paper could be extended to work for any integer $m$.
To this end, we would need to show that the projective graph
${\bf G}_{n,m}$, for a general integer $m$, has some kind of expanding property.
In fact, a generalized version of the tensor lemma (see Lemma \ref{lemma:tensor-general}) does exist for the
matrix $B_{n,m}$, and it is also possible to determine the
eigenvalues of $B_{n,m}$ for a general integer $m$ (see Theorems
\ref{theorem:eigenvalues-power} and \ref{theorem:eigenvalues-general}).

\begin{lemma}\label{lemma:tensor-general}
{\em (Zhang \cite{Z13}: Tensor Lemma)}
Let $m=m_1\cdots m_r=p_1^{e_1}\cdots p_r^{e_r}$ for distinct primes $p_1,\ldots, p_r$ and
positive integers
$e_1,\ldots, e_r$,
where $m_s=p_s^{e_s}$ for every $s\in [r]$.
Then we have that
\begin{equation}
B_{n,m}\simeq B_{n,m_1} \otimes \cdots \otimes B_{n,m_r}.
\end{equation}
\end{lemma}

\begin{theorem}\label{theorem:eigenvalues-power}
{\em (Zhang \cite{Z13}: Eigenvalues of $B_{n,m}$ when $m$ is a prime power)}
Let $m=p^e$, where $p$ is a prime  and $e$ is a positive integer, and let $n$ be a
positive integer. Then
$\lambda_1 = p^{2(e-1)(n-2)}\cdot \theta_{n-1,p}^2 $ is an eigenvalue of
$B_{n,m}$ of multiplicity $d_1=1$,
$\lambda_2 =   p^{(2e-1)(n-2)}$ is an eigenvalue of $B_{n,m}$ of multiplicity
$d_2 =  \theta_{n,p}-1$, and $\lambda_s=  p^{(2e+1-s)(n-2)}$
is an eigenvalue of $B_{n,m}$ of multiplicity $d_s=(p^{n-1}-1)\theta_{n,p^{s-2}}$
for every $s\in \{3,\ldots,e+1\}$.
\end{theorem}

\begin{theorem}\label{theorem:eigenvalues-general}
{\em (Zhang \cite{Z13}: Eigenvalues of $B_{n,m}$ when $m$ is any positive integer)}
Let $m=m_1\cdots m_r=p_1^{e_1}\cdots p_r^{e_r}$ for distinct primes $p_1,\ldots, p_r$ and
positive integers
$e_1,\ldots, e_r$,
where $m_s=p_s^{e_s}$ for every $s\in [r]$.
Let $\lambda_{s}$ be an eigenvalue of $B_{n,m_s}$ of
multiplicity $d_{s}$ for every
$s\in[r]$. Then $\lambda_1\cdots\lambda_r$
is an  eigenvalue of  $B_{n,m}$ of multiplicity
$d_{1}\cdots d_{r}$.
\end{theorem}

However, the method used in this paper
may become weaker as the number of distinct prime factors of $m$ increases.
As in many other classic applications,
 the performance of our method  depends on
the difference between the two largest eigenvalues of $B_{n,m}$.
Roughly speaking, a larger difference corresponds to better  performance.
However, Theorems \ref{theorem:eigenvalues-power} and \ref{theorem:eigenvalues-general}
show that this difference becomes less significant as the number of distinct prime factors
of $m$ increases. On the other hand, this does not rule out the possibility of
applying Theorems \ref{theorem:eigenvalues-power} and \ref{theorem:eigenvalues-general}
in a different way. Currently, an exponential gap still exists between the
best lower bound and the best upper bound for $k(m,n)$.  We hope that
these general theorems (Theorems \ref{theorem:eigenvalues-power} and
 \ref{theorem:eigenvalues-general})
can be used in some way to close this gap in the future.

\section*{Acknowledgment}

The authors thank the anonymous referees for very detailed and helpful comments that
significantly improved the organization and exposition of the paper.

The authors also gratefully acknowledge the support for this work received from
the National Research Foundation of Singapore under Research Grant NRF-CRP2-2007-03.

\ifCLASSOPTIONcaptionsoff
  \newpage
\fi

\begin{IEEEbiographynophoto}{Yeow Meng Chee}
received the B.Math. degree in computer science and combinatorics
and optimization and the M.Math. and Ph.D. degrees in computer science,
from the University of Waterloo, Waterloo, ON, Canada, in 1988, 1989,
and 1996, respectively.

Currently, he is an Associate Professor at the Division of Mathematical Sciences,
School of Physical and Mathematical Sciences, Nanyang Technological
University, Singapore. Prior to this, he was Program Director of Interactive Digital
Media R\&D in the Media Development Authority of Singapore, Postdoctoral
Fellow at the University of Waterloo and IBM's Zürich Research Laboratory,
General Manager of the Singapore Computer Emergency Response Team,
and Deputy Director of Strategic Programs at the Infocomm Development Authority,
Singapore. His research interest lies in the interplay between combinatorics
and computer science/engineering, particularly combinatorial design
theory, coding theory, extremal set systems, and electronic design automation.
\end{IEEEbiographynophoto}
\begin{IEEEbiographynophoto}{San Ling}
received the B.A. degree in mathematics from the University of Cambridge,
Cambridge, U.K., in 1985, and the Ph.D. degree in mathematics from
the University of California, Berkeley, in 1990. Since April 2005, he has been a
Professor with the Division of Mathematical Sciences, School of Physical and
Mathematical Sciences, Nanyang Technological University, Singapore. Prior to
that, he was with the Department of Mathematics, National University of Singapore.
His research fields include arithmetic of modular curves and application
of number theory to combinatorial designs, coding theory, cryptography, and
sequences.
\end{IEEEbiographynophoto}
\begin{IEEEbiographynophoto}{Huaxiong Wang}
obtained a Ph.D. in Mathematics from the University of Haifa,
Israel (1996) and a Ph.D. in Computer Science from the University of Wollongong,
Australia (2001). He is currently with Nanyang Technological University,
Singapore. His research interests include cryptography, information security,
coding theory, combinatorics and theoretical computer science. He is on the editorial
boards of Designs, Codes and Cryptography, Journal of Communications
and Journal of Communications and Networks and was the Program Co-Chair of
9th Australasian Conference on Information Security and Privacy (ACISP?4),
Sydney, Australia, July, 2004 and 4th International Conference on Cryptology
and Network Security (CANS05), Xiamen, Fujian, China, December, 2005. He
won the inaugural Prize of Research Contribution awarded by the Computer
Science Association of Australasia in 2004.
\end{IEEEbiographynophoto}
\begin{IEEEbiographynophoto}{Liang Feng Zhang}
received the B.S. degree in applied mathematics from the Tongji
University,
Shanghai, China, in 2004,
the M.S. degree in applied mathematics from the Shanghai Jiao Tong
University,
Shanghai, China, in 2007, and the Ph.D. degree in cryptography from
the Nanyang Technological University, Singapore, in 2012.
He is currently a postdoctoral fellow at the University of Calgary, Canada.
His research interests  include coding theory, cryptography, and discrete mathematics.
\end{IEEEbiographynophoto}

\end{document}